\documentclass[journal]{IEEEtran}

\ifCLASSINFOpdf
\else
\fi

\usepackage{graphicx} 
\usepackage{epstopdf} 
\usepackage{mathptmx} 
\usepackage{times} 
\usepackage{amsmath} 
\usepackage{amsthm}
\usepackage{amssymb}  
\usepackage{subfigure}
\usepackage{cite}
\usepackage{color}

\usepackage{algorithm}  
\usepackage{algorithmicx}  
\usepackage{algpseudocode}

\DeclareMathOperator{\rank}{rank}
\newtheorem{theorem}{Theorem}
\newtheorem{remark}{Remark}

\newcommand{\Rmnum}[1]{\uppercase\expandafter{\romannumeral #1}}  

\begin{document}

\title{An Adaptive Longitudinal Platooning \\
Design Based On Concurrent Learning
}

\author{Qiuhao Wen, Di Liu,~\IEEEmembership{Member,~IEEE}, Jiwei Wang, and Simone Baldi,~\IEEEmembership{Senior Member,~IEEE}
\vspace{-0.25cm}
\thanks{This research was supported by the European Union Horizon 2020
R\&I programme Marie Sklodowska-Curie grant 899987, by the Natural Science Foundation of China grants 62150610499 and 62073074 (corresponding authors: S. Baldi and D. Liu)\newline
Q. Wen and S. Baldi are with School of Mathematics, Southeast University, China {\tt\small s.baldi@tudelft.nl}\newline 
D. Liu is with School of Computation, Information and Technology, Technical University of Munich (TUM), Germany, and 
 with Visual Intelligence for Transportation lab, Ecole Polytechnique
Federale de Lausanne (EPFL), Switzerland {\tt\small di.liu@tum.de} \newline
J. Wang is with Bernoulli Institute for Mathematics, Computer Science \& Artificial Intelligence, University of Groningen, The Netherlands {\tt\small  jiwei.wang@rug.nl}
}
}

%
%

\IEEEtitleabstractindextext{
\begin{abstract}
    This work proposes a new adaptive longitudinal platooning strategy in the framework of concurrent learning. 
    Adaptive refers to vehicles facing uncertainty in powertrain parameters via on-line estimation; concurrent learning refers to using both current and past data in the estimation. 
    The proposed platooning strategy advances existing ones since convergence to the true powertrain parameters is guaranteed without imposing persistence of excitation on the vehicle behavior: it suffices the presence of a single non-zero data sample. 
    Meanwhile, the concurrent learning proof we give advances existing ones since it takes into account an extra unknown gain in the error dynamics.
\vspace{-0.1cm}
\end{abstract}
\begin{IEEEkeywords}
Longitudinal platooning, automated vehicles, adaptive control, concurrent learning. \vspace{-0.1cm}
\end{IEEEkeywords}
}

\maketitle

\thispagestyle{empty}
\pagestyle{empty}

\IEEEdisplaynontitleabstractindextext

%
\IEEEpeerreviewmaketitle

\section{Introduction}
Cooperative Adaptive Cruise Control (CACC) refers to a family of distributed control strategies that make use of {\color{black}on-board sensing (radar, tachometer, accelerometer, etc.) and inter-vehicle wireless communication} to establish longitudinal vehicle platoons with desired inter-vehicle distance~\cite{c1}. 
{\color{black}In CACC, the terms `adaptive' and `control' are not used in the sense of adaptive control theory,
but in the sense of handling different traffic/speed regimes, 
e.g.,  mixed traffic with automated and human-driven vehicles~\cite{c3,c4,c5}.
Actually, most CACC strategies are not designed based on adaptive control theory: 
uncertainties stemming from vehicle dynamics have been traditionally addressed in CACC in a robust control sense~\cite{c6,c7}, 
while uncertainties stemming from unreliable wireless communication have been addressed in a stochastic~\cite{c8,c9} or switched control sense~\cite{c10,c12}.} 

{\color{black}
A recent new perspective to CACC has shown that the distributed control setting~\cite{c13,c14} reduces to a decentralized one when some disturbance decoupling properties are enforced~\cite{c15,c16}. 
The interest of such perspective is that disturbance decoupling automatically guarantees other properties traditionally sought in platooning, such as string stability (i.e., rejection of noises propagating through the platoon). 
}
{\color{black}
CACC strategies truly based on adaptive control theory, e.g., Model Reference Adaptive Control (MRAC)~\cite{c17,c18}, or learning methods~\cite{c19,c20}, 
embed estimation mechanisms in the CACC protocol, so as to estimate the uncertainty on-line. 
}
The most typical uncertainty is the time constant of the powertrain, which {\color{black}in turn} is reflected into uncertainty of the control gains.
{\color{black}
When applying adaptive and learning tools into CACC, a natural question is whether estimated uncertainty will converge to its actual value. 
While this typically requires Persistence of Excitation (PE) conditions~\cite{c30}, 
several studies have been done to relax PE conditions,} one of the first being the Concurrent Learning (CL) framework~\cite{c21}. 
In concurrent learning, adaptation relies on both current and past data, so that PE conditions are relaxed to rank conditions on data. 
The framework is general enough to be applicable to switched systems~\cite{c22}, time-varying systems~\cite{c23,c24}, and classes of nonlinear systems~\cite{c25,c26,c27}.

{\color{black}Relaxed PE conditions have attracted increasing interest and studies in this sense include initial/finite excitation~\cite{c31,c32}, 
dynamic regressor extension and mixing~\cite{c33},} 
and data informativity~\cite{c34}{\color{black};} 
while each one of these frameworks is interesting per se, in this work we propose a new CACC strategy in the framework of {\color{black}CL}. 
In particular, we focus on the disturbance decoupling perspective of~\cite{c15,c16}, {\color{black}motivated by recent CACC designs able to attain disturbance decoupling using adaptive control~\cite{c28,c29}.} 
{\color{black}The main contribution of this work is to prove CL in CACC applications with disturbance decoupling.} 
{\color{black}
We advance the state-of-the-art CACC in~\cite{c28,c29} by reformulating the adaptive controller as indirect MRAC,
to estimate a single parameter (the powertrain time constant) instead of four control gains.
While a single estimate already decreases the requirements for PE, we prove that PE conditions can be removed using CL.}
A single non-zero data sample is enough  to guarantee convergence to the true powertrain parameter. 
{\color{black}The CL proof we give advances available ones~\cite{c21}, since it takes into account an extra unknown gain in the error dynamics.} 

The rest of the paper is organized as follows:
{\color{black}Section \Rmnum{2} gives the problem formulation.}
A platooning strategy based on {\color{black}indirect} MRAC is formulated in Section \Rmnum{3};
{\color{black}PE is overcome in Section \Rmnum{4} by CL.}
{\color{black}
Section \Rmnum{5} contains simulations, while Section \Rmnum{6} presents conclusions.
}
\section{Control Problem Statement}
For a pair of predecessor-follower vehicles, indexed as $\{p, f\}$, longitudinal platooning dynamics are usually expressed in the literature~\cite{c5,c10,c14,c15} as follows
\begin{equation}\label{standard dyanamics}
    \begin{aligned}
        \dot{s}_i(t)&=v_i(t),\\
        \dot{v}_i(t)&=a_i(t),\\
        \tau_i\dot{a}_i(t)&=-a_i(t)+u_i(t),
    \end{aligned}
    \quad i\in\{p,f\},
\end{equation}
where $s_i, v_i, a_i$ represent the longitudinal position, velocity, acceleration of vehicle $i\in \{p,f\}$, 
and $u_i$ represents the desired acceleration that enters through a first-order system with time constant {\color{black}$\tau_i>0$ that \emph{will be assumed unknown}}. 
{\color{black}This first-order dynamics represents the powertrain, $\tau_i$ being the time constant for reaching the desired acceleration.} 

The longitudinal platooning task is established through the spacing error, defined as follows.
Let
\begin{equation}
    d_f(t)=s_p(t)-s_f(t)\label{real distance}
\end{equation}
denote the inter-vehicle distance between $f$ and $p$, and let
\begin{equation}
    d_f^*(t)=hv_f(t)+r\label{desired distance}
\end{equation}
denote a desired inter-vehicle spacing policy, with $h>0$ being the time headway and $r \geq 0$ a standstill distance. 
The desired spacing policy (\ref{desired distance}) represents a typical driving behavior of increasing the inter-vehicle distance as the velocity increases~\cite{c1,c4}. 
Without loss of generality, we take $r=0$, as the standstill distance can be removed by coordinate shift.
The spacing error is
\begin{equation}
    e_f(t)=d_f(t)-d_f^*(t)=s_p(t)-s_f(t)-hv_f(t).\label{spacing error}
\end{equation}

{\color{black}Let us also define $\nu_f(t) = v_p(t)-v_f(t)$ as the relative velocity between the two vehicles.} 
The control problem is formulated as follows.

\textbf{Control Problem:}
    Consider the predecessor-follower model (\ref{standard dyanamics}) with desired spacing policy (\ref{desired distance}).
    The control problem is to design an adaptive controller $u_f$ such that, 
    for any unknown $\tau_i$,\ $i \in \{p,f\}$, the following objectives are achieved:
    \begin{enumerate}
        \item For any bounded $u_p(\cdot)$, it holds that
        \begin{equation}
            \lim_{t\rightarrow\infty}e_f(t)=0;\label{closed loop stability}
        \end{equation} 
        \item For any bounded $u_p(\cdot)$, it holds that
        \begin{equation}
            \lim_{t \rightarrow \infty} \int_t^{t+T}(\lvert v_f(\ell) \rvert ^2 -\lvert v_p(\ell) \rvert ^2) d \ell \le 0,\quad \forall T>0.\label{closed loop string stability}
        \end{equation}
    \end{enumerate}
\begin{remark}[Meaning of (\ref{closed loop stability})-(\ref{closed loop string stability})]
    As shown in~\cite{c15,c16}, {\color{black}item 1)} in the Control Problem guarantees disturbance decoupling and output stability with $u_p$ acting as a disturbance.
    {\color{black}Item 2) is a way to represent string stability,} i.e., velocity perturbations should not amplify from the predecessor to the follower vehicle.
\end{remark}
\begin{remark}[Disturbance decoupling implies string stability]
    As shown in~\cite{c15,c16} and later in~\cite{c28}, disturbance decoupling automatically implies string stability, 
    that is, {\color{black}item 1) implies item 2).} 
    Therefore, in the rest of the work we focus on item 1) for compactness and without loss of generality.
\end{remark}

To represent the dynamics of the predecessor-follower system, let us define $\xi^\top=[\xi^\top_p \ \xi^\top_f]$, where $\xi^\top_i=[s_i \ v_i \ a_i]\in\mathbb{R}^3,\ i\in\{p,f\}$.
From (\ref{standard dyanamics}), we obtain the dynamics
\begin{equation}
    \dot{\xi}(t)=\begin{bmatrix}
        \bar{A}_p & 0\\
        0 & \bar{A}_f
    \end{bmatrix}\xi(t)+\begin{bmatrix}
        0\\
        \bar{B}_f
    \end{bmatrix}u_f(t)+\begin{bmatrix}
        \bar{B}_p\\
        0
    \end{bmatrix}u_p(t),\label{cruise control open loop}
\end{equation}
with matrices $\bar{A}_i,\ \bar{B}_i$ given by
\[
    \bar{A}_i=\begin{bmatrix}
        0 & 1 & 0\\
        0 & 0 & 1\\
        0 & 0 & -\tau_i^{-1} 
    \end{bmatrix},\
    \bar{B}_i=\begin{bmatrix}
        0\\
        0\\
        \tau_i^{-1}
    \end{bmatrix},\
    i\in\{p,f\}.
\]

Yet, {\color{black}a more convenient representation of the predecessor-follower system is by defining $x = [e_f\ \nu_f\ a_f]^\top \in \mathbb{R}^3$, leading to}
\begin{equation}
    \dot{x}(t)=A x(t) + B u_f(t) + G a_p(t),
    \label{lower dimensional closed loop systems}
\end{equation}
with 
\begin{equation}
    A=\begin{bmatrix}
        0 &1 &-h\\
        0 &0 &-1\\
        0 &0 &-\tau_f^{-1}
    \end{bmatrix},
    B=\begin{bmatrix}
        0\\
        0\\
        \tau_f^{-1}
    \end{bmatrix},
    G=\begin{bmatrix}
        0\\
        1\\
        0
    \end{bmatrix}.
    \label{actual model matrices}
\end{equation}
The advantage of $x$ in (\ref{lower dimensional closed loop systems}) as compared to $\xi$ in (\ref{cruise control open loop}) is to give a lower dimensional representation of the platooning system. 
{\color{black}Note that the exogenous input is $u_p$ in (\ref{cruise control open loop}) and $a_p$ in (\ref{lower dimensional closed loop systems}).} 
This is without loss of generality from the point of view of the Control Problem: in fact, it can be inferred from the first-order system in (\ref{standard dyanamics}) that any bounded $u_p(\cdot)$ results in a bounded $a_p(\cdot)$.

{\color{black}
If $\tau_i$ were known,~\cite{c28} has shown that the controller solving the Control Problem would be linear, $u_f(t)=F\xi(t)$ with
\begin{equation}
    F=[\theta_1 \ \ \theta_2 \ \ \tau_f h^{-1} \ \ -\hspace{-0.075cm}\theta_1 \ \ -\hspace{-0.075cm}\theta_2\hspace{-0.075cm}-\hspace{-0.075cm}h\theta_1 \ \ 1\hspace{-0.075cm}-\hspace{-0.075cm}\tau_f h^{-1}\hspace{-0.1cm}-\hspace{-0.075cm}h\theta_2],
    \label{input matrix}
\end{equation}
with arbitrary $\theta_1,\ \theta_2>0$.}
Equivalently, one can write the same controller as $u_f(t)=K  x(t)+L a_p(t)$ with
\begin{equation}
    K = [\theta_1\ \ \theta_2\ \ 1\hspace{-0.075cm}-\hspace{-0.075cm}\tau_f h^{-1}\hspace{-0.1cm}-\hspace{-0.075cm}h\theta_2], \quad L = \tau_f h^{-1},
    \label{controller gains}
\end{equation}
containing a feedback term and a feedforward term. 
Such a formulation is useful for MRAC, 
where the control action indeed contains feedback and feedforward terms~\cite{c30}.

{\color{black}Being the powertrain time constant $\tau_f$ unknown, the linear controller (\ref{input matrix}) (equivalently, (\ref{controller gains})), cannot be implemented.} 
This is why we next look for an adaptive controller in place of a linear one. 

\section{Indirect MRAC Design for Platooning}
{\color{black}We will present an indirect MRAC design as the basis of the strategy to be proposed. 
The MRAC design will be illustrated in next three steps.}
\subsection{First Step: Design of a Reference Model}
A suitable reference model is obtained upon defining a `virtual' vehicle, which will be indexed by $\bar{f}$.
Clearly, vehicle $\bar{f}$ should follow similar dynamics as (\ref{standard dyanamics}), i.e.,
\begin{equation}
    \begin{aligned}
        \dot{s}_{\bar{f}}(t)&=v_{\bar{f}}(t),\\
        \dot{v}_{\bar{f}}(t)&=a_{\bar{f}}(t),\\
        \tau_{\bar{f}} \dot{a}_{\bar{f}}(t)&=-a_{\bar{f}}(t)+u_{\bar{f}}(t),
    \end{aligned}
    \label{virtual dyanamics}
\end{equation}
where $\tau_{\bar{f}}>0$ is a nominal powertrain time constant chosen by the designer.
Then, consider the controller
\begin{equation}
    u_{\bar{f}}(t)\hspace{-0.075cm}=\hspace{-0.075cm}\theta_1 e_{\bar{f}}(t) + \theta_2 \nu_{\bar{f}}(t) + (1\hspace{-0.075cm}-\hspace{-0.075cm}\tau_{\bar{f}}h^{-1}\hspace{-0.1cm}-\hspace{-0.075cm}h\theta_2)a_{\bar{f}}(t) +\tau_{\bar{f}}h^{-1}a_p(t),
    \label{reference model input}
\end{equation}
{\color{black}sharing the same form as (\ref{input matrix}),} with $\tau_{\bar{f}}$ in place of $\tau_{f}$.

By defining $\bar{x} = [e_{\bar{f}}\ \nu_{\bar{f}}\ a_{\bar{f}}]^\top$ and closing the loop with (\ref{reference model input}) and (\ref{virtual dyanamics}), we obtain the reference model
\begin{equation}
    \dot{\bar{x}}(t) =\bar{A}\bar{x}(t) + \bar{G}a_p(t),
    \label{reference model}
\end{equation}
where
\begin{equation}
    \bar{A}=\begin{bmatrix}
        0 &1 &-h\\
        0 &0 &-1\\
        \theta_1\tau_{\bar{f}}^{-1} &\theta_2\tau_{\bar{f}}^{-1} &-h^{-1}\hspace{-0.075cm}-\hspace{-0.075cm}h\theta_2\tau_{\bar{f}}^{-1}
    \end{bmatrix},
    \bar{G}=\begin{bmatrix}
        0\\
        1\\
        h^{-1}
    \end{bmatrix}.
    \label{reference model matrices}
\end{equation}
Note that $\bar{A}$ is Hurwitz, cf.~\cite{c28}.
Because (\ref{reference model input}) has the same form as (\ref{input matrix}), it satisfies disturbance decoupling and string stability (cf. Control Problem) for the reference model (\ref{reference model}).

\subsection{Second Step: Design of a Model-Matching Control}
{\color{black}To elaborate the process of control design, let us start with an ideal case, i.e., $\tau_f$ is known.
In this case, consider} 
\begin{equation}
    u_f^*(t) = a_f(t) + \tau_f\phi(x(t),a_p(t)),
    \label{ideal controller}
\end{equation}
where $u_f^*$ stands for an ideal version of $u_f$, and $\phi(x,a_p)=\bar{K} x + h^{-1} a_p$ with $\bar{K}=[\theta_1\tau_{\bar{f}}^{-1}\ \theta_2\tau_{\bar{f}}^{-1}\ -\hspace{-0.075cm}(h^{-1}\hspace{-0.075cm}+\hspace{-0.075cm}h\theta_2\tau_{\bar{f}}^{-1})]$.
In the following, for brevity, we will denote $\phi(x(t),a_p(t))$ as $\phi(t)$.
The controller (\ref{ideal controller}) is ideal because by substituting it into (\ref{lower dimensional closed loop systems}), together with (\ref{reference model}) and (\ref{reference model matrices}), gives
\begin{equation}
    {\color{black} \dot{\tilde{x}}(t)=\bar{A}\tilde{x}(t),}
    \label{ideal tracking error}
\end{equation}
where $\tilde{x}=x-\bar{x}$ is the tracking error between the closed-loop system and the reference model. 
In (\ref{ideal tracking error}), $\tilde{x}$ converges to zero as $\bar{A}$ is Hurwitz,
which means that the closed-loop system converges to the same behavior of the reference model.

{\color{black}As $\tau_f$ in (\ref{actual controller}) is unknown,} let us propose, instead of (\ref{ideal controller}), an adaptive controller given by
\begin{equation}
    u_f(t)=a_f(t)+\hat{\tau}_f(t)\phi(t),
    \label{actual controller}
\end{equation}
with $\hat{\tau}_f$ being an estimate of $\tau_f$: the problem is to design an appropriate adaptive law for $\hat{\tau}_f$.
To this purpose, let us calculate the dynamics of $\tilde{x}$ resulting from (\ref{lower dimensional closed loop systems}) and (\ref{reference model}) with controller (\ref{actual controller}), that is
\begin{equation}
    \dot{\tilde{x}}(t)=\bar{A}\tilde{x}(t)+\tilde{B}\frac{\tilde{\tau}_f(t)}{\tau_f}\phi(t),
    \label{actual tracking error}
\end{equation}
where $\tilde{\tau}_f(t)=\hat{\tau}_f(t)-\tau_f$ is the estimation error and $\tilde{B}^\top=[ 0\ 0 \ 1]$.
Note that (\ref{actual tracking error}) contains the same ideal dynamics of (\ref{ideal tracking error}), plus a disturbance term given by the estimation error. 

\subsection{Third Step: Standard MRAC Adaptive Law}
The next step is to design the adaptive law for $\hat{\tau}_f$ such that $\tilde{x}$ in (\ref{actual tracking error}) can still converge to zero. 
In the following, we will omit time dependence when obvious.
An adaptive law to estimate the unknown powertrain parameter is proposed as
\begin{equation}
    \dot{\hat{\tau}}_f=-\gamma_1 \tilde{B}^\top P\tilde{x}\phi,
    \label{adaptive law}
\end{equation}
where $\gamma_1$ is a positive learning gain
and $P=P^\top$ is the positive definite solution of the Lyapunov equation $\bar{A}^\top P+P\bar{A}+Q=0$ with $Q>0$ an arbitrary positive definite matrix.

The following stability result holds.
\begin{theorem}
    Consider the closed-loop system formed by
    the vehicle dynamics (\ref{lower dimensional closed loop systems}), the spacing policy (\ref{spacing error}), and the reference
    model (\ref{reference model}). The controller (\ref{actual controller}) with adaptive law (\ref{adaptive law}), achieves (\ref{closed loop stability}).
\end{theorem}
\begin{proof}
    Consider the Lyapunov function candidate
    \begin{equation}\label{Lyacan}
        V(\tilde{x},\tilde{\tau}_f)=\frac{1}{2}\tilde{x}^\top P \tilde{x}+\frac{\tilde{\tau}_f^2}{2\gamma_1\tau_f}.
    \end{equation}

    Differentiating (\ref{Lyacan}) along the trajectory of (\ref{actual tracking error}) with the adaptive law (\ref{adaptive law}) {\color{black}gives}
    \begin{equation}\label{Lycande}
            \dot{V}(\tilde{x},\tilde{\tau}_f) = -\frac{1}{2}\tilde{x}^\top Q \tilde{x}\le0,
    \end{equation}
    {\color{black}indicating} that the origin $(\tilde{x},\tilde{\tau}_f)=(0,0)$ is uniformly stable. 
    This implies that $\tilde{x}(\cdot),\ \tilde{\tau}_f (\cdot)$ are bounded.
    
    Now we use Barbalat's lemma to {\color{black}show} convergence of $\tilde{x}$: the lemma states that
    if a signal $g(\cdot)$ and its time derivative satisfy $g, \dot{g}\in\mathcal{L}_{\infty}$ and $g\in\mathcal{L}_2$,
    then $g(t)\rightarrow0$ as $t\rightarrow\infty$. 
    {\color{black}Recall} that $\mathcal{L}_2$ and $\mathcal{L}_{\infty}$ indicate the sets of signals bounded in the $\mathcal{L}_2$ and $\mathcal{L}_{\infty}$ norms, respectively.  

   One can obtain from (\ref{standard dyanamics}) that $a_p\in\mathcal{L}_{\infty}$, because $u_p(\cdot)$ is assumed to be bounded. 
    From (\ref{reference model}), we have $\bar{x}\in\mathcal{L}_\infty$, which implies $x=\bar{x}+\tilde{x}\in\mathcal{L}_\infty$.
    As a result, one can conclude that $\dot{\tilde{x}}\in\mathcal{L}_\infty$ from (\ref{actual tracking error}).
    To prove $\tilde{x}\in\mathcal{L}_2$, integrate (\ref{Lycande})
    \begin{equation}
        \frac{1}{2}\int_{0}^{\infty} {\tilde{x}(t)}^\top Q \tilde{x}(t) dt = V(0)-V_\infty,
        \label{Lycandeint}
    \end{equation}
    where $V_\infty=\lim_{t\rightarrow\infty}V(t)$ is bounded as $\dot{V}\le0$ and this implies that $\tilde{x}\in\mathcal{L}_2$. 
    Hence, through Barbalat's lemma, one can obtain $\tilde{x}\rightarrow0$ as $t\rightarrow\infty$.
    Since $x=\bar{x}+\tilde{x}$, we have
    \begin{equation*}
            \lim_{t\rightarrow\infty}e_f(t)=\lim_{t\rightarrow\infty}\bar{H}\bar{x}(t)+\lim_{t\rightarrow\infty}\bar{H}\tilde{x}(t)=0,
    \end{equation*}
    where $\bar{H}=[1\ 0\ 0]$.
    This ends the proof.
\end{proof}

\section{Concurrent Learning Design for Platooning}
As standard in MRAC, Theorem 1 does not guarantee convergence of $\hat{\tau}_f$ to $\tau_f$. 
With PE conditions, such convergence is guaranteed~\cite{c30}. 
{\color{black}For the scalar $\phi$ in Theorem 1, imposing PE could still be restrictive in practice: a platoon converging to constant velocity would converge to zero acceleration and fail to meet PE.} 
To address this problem, we extend MRAC in the sense of CL.

The proposed MRAC-CL design is
\begin{equation}
    \dot{\hat{\tau}}_f=-\gamma_1 \tilde{B}^\top P\tilde{x}\phi-\sum_{j=0}^{k}\gamma_2 \phi_j\bar{\epsilon}_j,
    \label{adaptive law CL}
\end{equation}
where $\gamma_2$ is another positive learning gain,
$\phi_j=\phi(x_j,{a_p}_j)$ with {\color{black}$j\in\{0,1,\dots,k\}$} indexing stored data points, and $\bar{\epsilon}_j$ calculated according to the following remark.
\begin{remark}[{\color{black}Past samples}]
    Let us rewrite (\ref{lower dimensional closed loop systems}) as
    {\color{black} $\dot{x}=\bar{A} x+\bar{G} a_p+\tilde{B}\tau_f^{-1}\epsilon$}, with $\epsilon =u_f-u_f^*$, 
    and define $\bar{\epsilon}=\tau_f^{-1} \epsilon$. 
    We have {\color{black}$\bar{\epsilon}=\tilde{B}^{\dag}(\dot{x}-\bar{A} x-\bar{G} a_p)$},
    with {\color{black}$\tilde{B}^\dag={(\tilde{B}^\top \tilde{B})}^{-1}\tilde{B}^\top$} the pseudo-inverse of $\tilde{B}$. 
    {\color{black}Being the matrices $\bar{A}$, $\bar{G}$, $\tilde{B}$ known, $\bar{\epsilon}_j$ can be obtained from past samples of the signals needed to calculate $\bar{\epsilon}$: this is in line with the collection of past samples in standard CL~\cite{c21,c22}.
    An integral CL was proposed in~\cite{c35} where the need for higher order derivatives as in~\cite{c21,c22} is avoided.}
\end{remark}
The following stability and convergence result holds.
\begin{theorem}\label{main stability theorem}
    Consider the closed-loop system formed by the vehicle dynamics (\ref{lower dimensional closed loop systems}), the spacing policy (\ref{spacing error}) and the reference model (\ref{reference model}).
    The controller (\ref{actual controller}) with adaptive law (\ref{adaptive law CL}), achieves (\ref{closed loop stability}). 
    Furthermore, if $\rank[\phi_0\ \phi_1\ \cdots\ \phi_k]=1$, i.e., 
    there exists at least one sample $\phi_j \neq 0,\ {\color{black}j\in\{0,1,\dots,k\}}$,
    then $\hat{\tau}_f\rightarrow\tau_f$ as $t\rightarrow\infty$.
\end{theorem}
\begin{proof}
    Consider the Lyapunov function candidate (\ref{Lyacan}) and let $\eta=\Vert\tilde{x}\Vert^2+\tilde{\tau}_f^2$. 
    Then, we have
    \begin{align*}
        \frac{1}{2}{\min}(\lambda_{\min}(P),(\gamma_1\tau_f)^{-1})\eta &\leq V(\tilde{x},\tilde{\tau}_f)\\ 
        &\leq \frac{1}{2}{\max}(\lambda_{\max}(P),(\gamma_1\tau_f)^{-1})\eta,
    \end{align*}
    where $\lambda_{\min}(\cdot), \lambda_{\max}(\cdot)$ denote the smallest and the largest eigenvalue.
    Differentiating (\ref{Lyacan}) along the trajectory of (\ref{actual tracking error}) with the adaptive law (\ref{adaptive law CL}) {\color{black}results in}
    \begin{equation*}
            \dot{V}(\tilde{x},\tilde{\tau}_f) = -\frac{1}{2}\tilde{x}^\top Q \tilde{x} - \frac{\gamma_2}{\gamma_1\tau_f^2}\tilde{\tau}_f^2\sum_{j=0}^{k}\phi_j^2.
    \end{equation*}
		
\begin{figure}[t] 
	\centering
	\includegraphics[width=8.9cm]{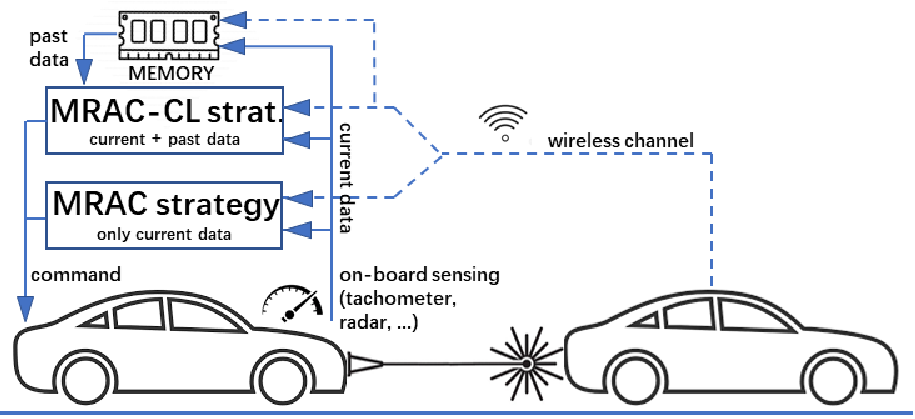}\vspace{-0.15cm}  
	\caption{{\footnotesize {\sf In CACC based on MRAC, only current data are used for control and adaptation. CACC based on MRAC-CL also uses past data.}\vspace{-0.45cm}}} 
	\label{CACCdraw}
\end{figure}
\begin{figure*}[thpb]\label{error and velocity} 
    \centering
    \subfigure[Non-adaptive control (\ref{input matrix}) when $u_0$ guarantees PE.]{                           
        \includegraphics[width=0.45\linewidth]{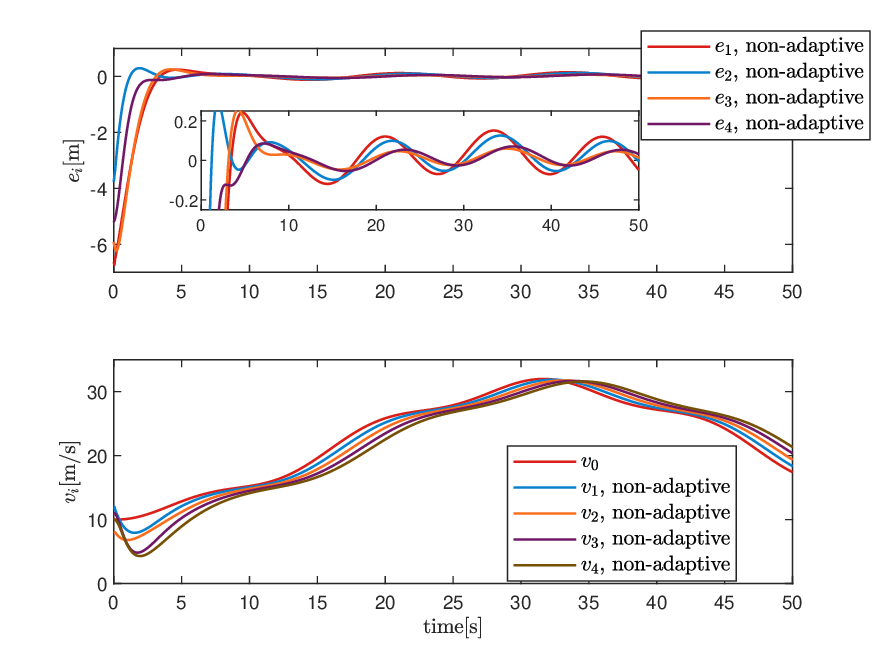}} 
    \subfigure[MRAC (\ref{actual controller}), (\ref{adaptive law}) and MRAC-CL (\ref{actual controller}), (\ref{adaptive law CL}) when $u_0$ guarantees PE.]{                            
        \includegraphics[width=0.45\linewidth]{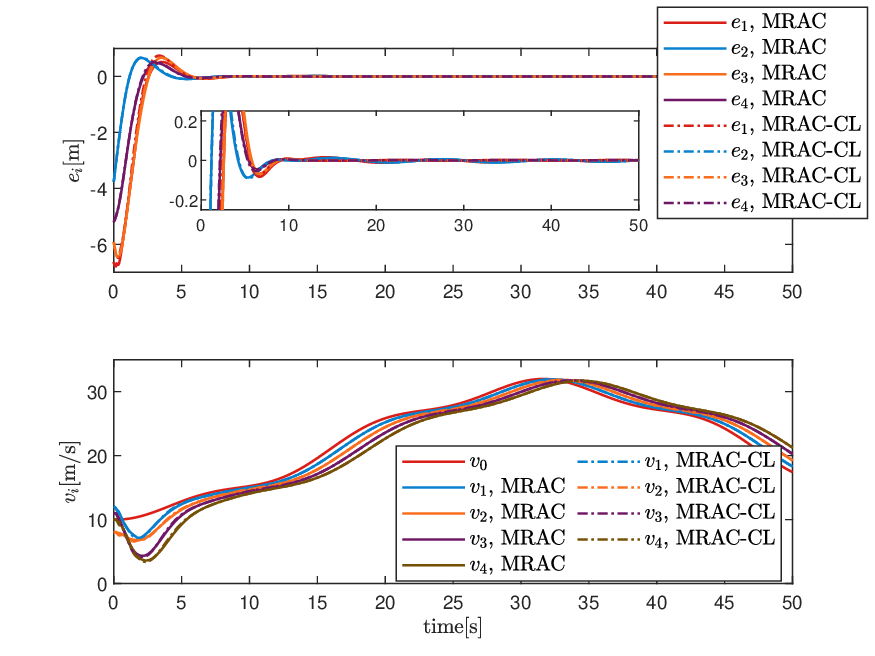}}  

    \subfigure[Non-adaptive control (\ref{input matrix}) when $u_0$ does \emph{not} guarantee PE.]{                        
        \includegraphics[width=0.435\linewidth]{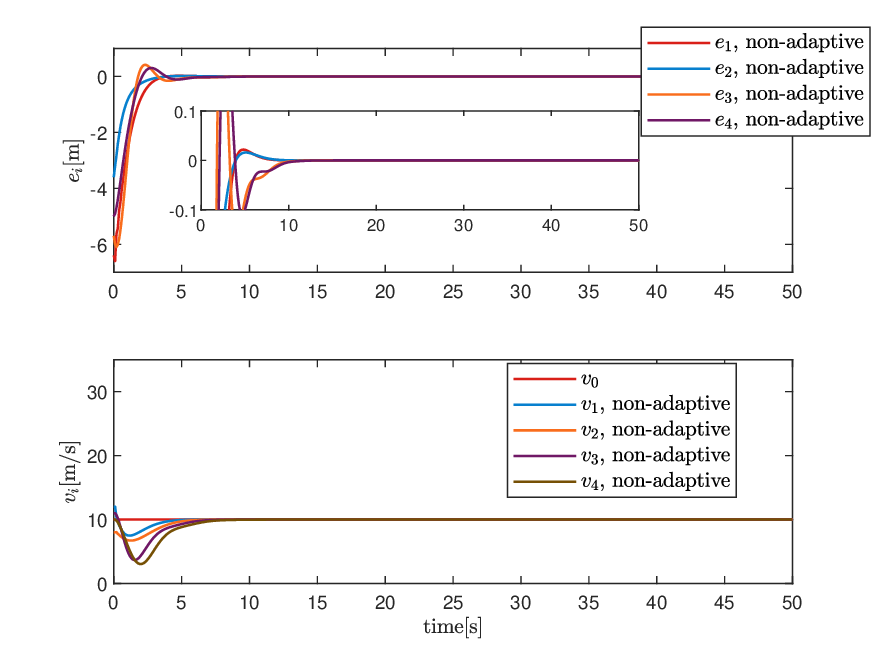}}
    \subfigure[MRAC (\ref{actual controller}), (\ref{adaptive law}) and MRAC-CL (\ref{actual controller}), (\ref{adaptive law CL}) when $u_0$ \newline does \emph{not} guarantee PE.]{                           
        \includegraphics[width=0.435\linewidth]{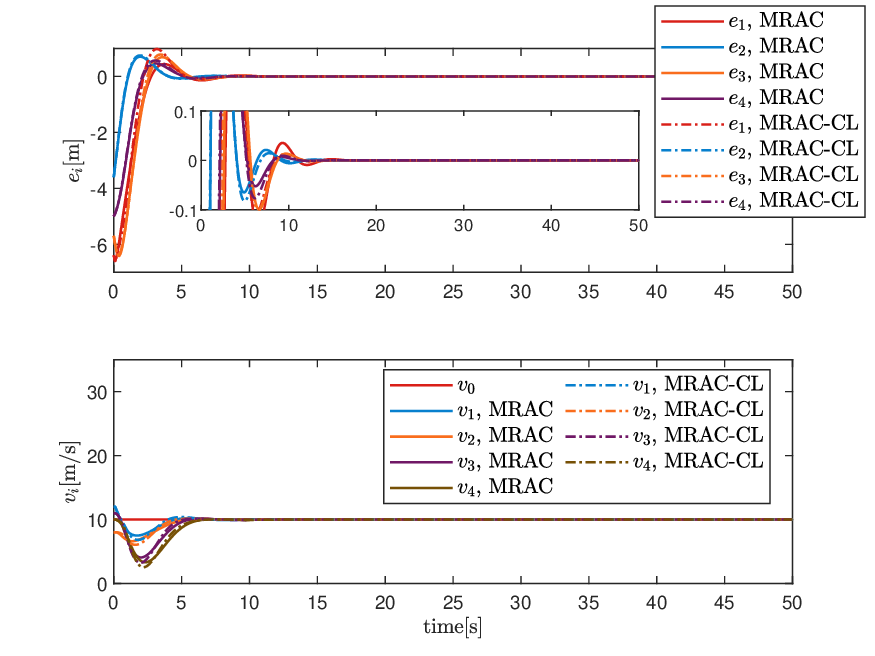}} 
    \caption{Spacing errors and velocities with different controllers (non-adaptive,  adaptive). 
		When $u_0$ guarantees PE, the spacing errors for non-adaptive control oscillate, indicating no disturbance decoupling. When $u_0$ does \emph{not} guarantee PE, all spacing errors converge due to lack of disturbance.}
\end{figure*}

  If there exists at least one sample $\phi_j \neq 0,\ {\color{black}j\in\{0,1,\dots,k\}}$, 
  we have $\Omega=\sum_{j=0}^{k}{\phi_j}^2>0$ 
  and $\dot{V}(\tilde{x},\tilde{\tau}_f)\leq-\frac{1}{2}\lambda_{\min}(Q)\tilde{x}^\top \tilde{x}-\Omega\frac{\gamma_2}{\gamma_1{\tau_f}^2}\tilde{\tau}_f^2$.
  As a result, we obtain
  \begin{equation}\label{derivative of CL}
    \color{black}\dot{V}(\tilde{x},\tilde{\tau}_f)\leq-\frac{{\max}(\lambda_{\min}(Q),2\Omega\gamma_2{(\gamma_1{\tau_f}^2)}^{-1})}{{\min}(\lambda_{\min}(P),{(\gamma_1\tau_f)}^{-1})}V(\tilde{x},\tilde{\tau}_f).  
  \end{equation}
  In other words, $V(\tilde{x},\tilde{\tau}_f)$ satisfies the requirements of the Lyapunov Stability Theorem, meaning that $\tilde{x},\tilde{\tau}_f \rightarrow 0$ asymptotically. 
  This ends the proof.
\end{proof}
\begin{remark}[Contribution of Theorem 2]
    {\color{black}Theorem 2 is able to handle the unknown gain $\tau_f^{-1}$ in the error dynamics (\ref{actual tracking error}): 
    this gain is neglected in existing CL~\cite{c21,c35}.
    This was possible by including $\tau_f$ in (\ref{Lyacan}) and by modifying the definition of $\bar{\epsilon}$ in (\ref{adaptive law CL}) accordingly (cf. Remark 3): although unknown, $\tau_f$ can appear in (\ref{Lyacan}) since the Lyapunov
    function is used only for stability analysis.}
    \end{remark}
{\color{black}
    Although sharing the same Lyapunov function~(\ref{Lyacan}), 
    Theorem 1 can only guarantee simple Lyapunov stability of the origin, 
    i.e., $\tilde{\tau}_f(\cdot)$ can only be proven a bounded signal; 
    meanwhile, Theorem 2 guarantees asymptotic stability of the origin, implying convergence of $\tilde{\tau}_f$ to zero.
    This is possible thanks to the additional term $-\sum_{j=0}^{k}\gamma_2 \phi_j\bar{\epsilon}_j$ in (\ref{adaptive law CL}). 
    Let us recall from the literature that simple Lyapunov stability at the origin is one of the main reasons for lack of robustness margins in adaptive loops~\cite[Sect. 3.11]{c30}, whereas convergence to the true parameters guarantee robustness margins.}

{\color{black}Let us summarize the steps for implementing the proposed platooning strategy based on MRAC-CL.}  
    \begin{algorithm}
        \renewcommand{\thealgorithm}{}
        {\color{black}
        \floatname{algorithm}{\color{black}Platooning strategy based on MRAC-CL}
        \caption{}
        \begin{algorithmic}[1]
        \Require nominal powertrain constant $\tau_{\bar{f}}$, positive definite $Q$, update gains $\gamma_1,\ \gamma_2$, stored data points $\bar{\epsilon}_j,\ \phi_j,\ j\in\{0,1,\dots,k\}$, current state $x(t) = [e_f(t)\ \nu_f(t)\ a_f(t) ]^\top$ from on-board sensors, communicated signal $a_p(t)$. 
        \Ensure control $u_f(t)$, estimate $\hat{\tau}_f(t)$.  
        \vspace{0.15cm}
        \State Run the reference model (\ref{reference model}) with controller (\ref{reference model input}) to obtain the reference state $\bar{x}(t)$;
        \State Form $\phi(x(t),a_p(t))$ based on the current $x(t),\ a_p(t)$, and run the adaptive controller (\ref{actual controller});
        \State Calculate the tracking error $\tilde{x}(t)=x(t)-\bar{x}(t)$, and update $\hat{\tau}_f$ via (\ref{adaptive law CL}).
        \end{algorithmic}
        }
    \end{algorithm}

\section{Simulations}
\begin{figure*}[thb]
\centering
    \subfigure[MRAC (\ref{actual controller}), (\ref{adaptive law}) and MRAC-CL (\ref{actual controller}), (\ref{adaptive law CL}) when $u_0$ \newline guarantees PE.]{
        \includegraphics[width=0.435\linewidth]{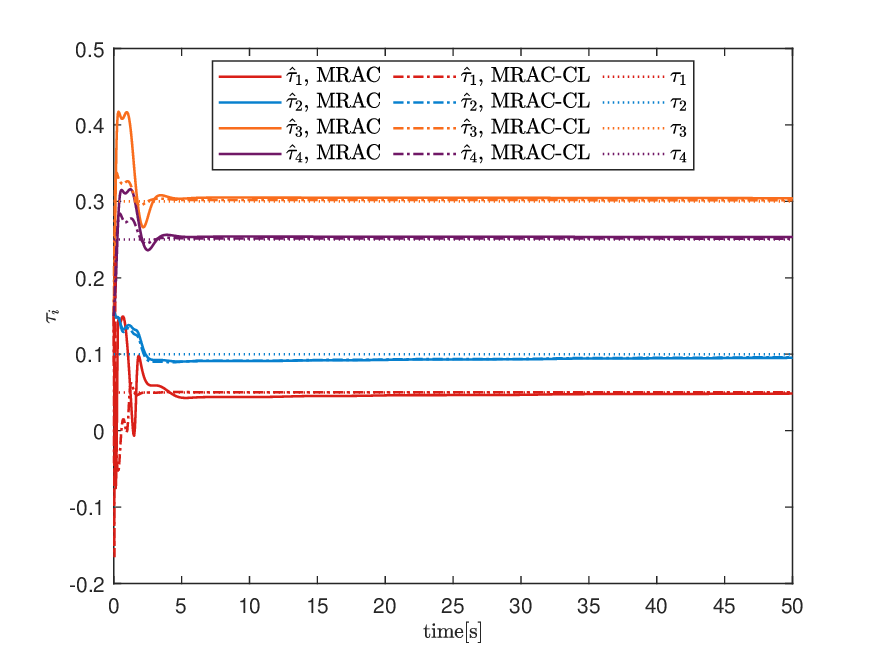}}   
    \subfigure[MRAC (\ref{actual controller}), (\ref{adaptive law}) and MRAC-CL (\ref{actual controller}), (\ref{adaptive law CL}) when $u_0$ \newline does \emph{not} guarantee PE.]{
        \includegraphics[width=0.435\linewidth]{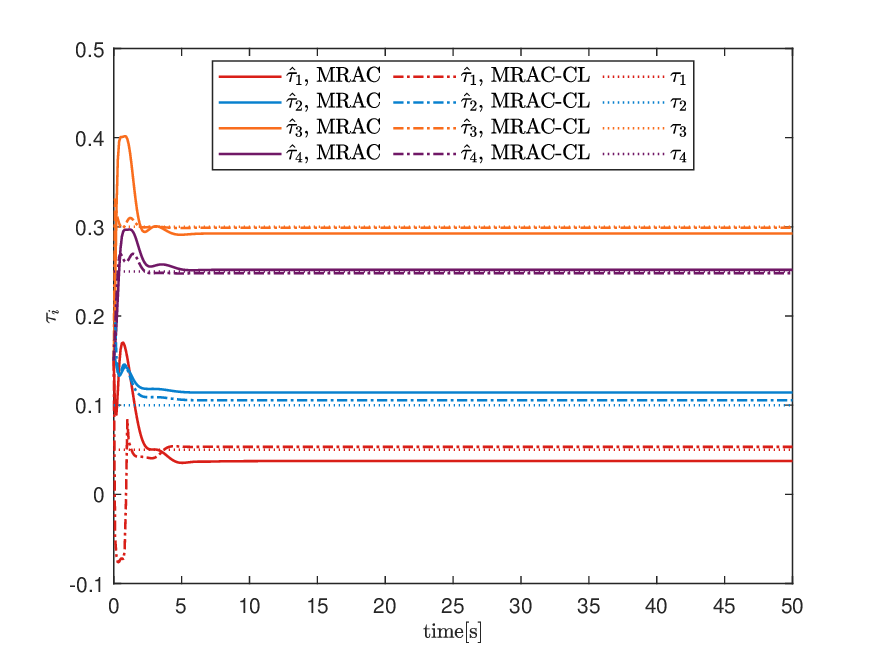}}
    \caption{Parameter convergence of $\hat{\tau}_{i}$ with different adaptive controllers: with PE, both controllers converge to the actual parameters, but without PE only MRAC-CL achieves convergence.}
    \label{tau}
\end{figure*}
To verify the theoretical analysis, this section presents simulation results of a platoon with five vehicles (one leader indexed as $0$ and four following vehicles indexed as $1,\ 2,\ 3,\ 4$)\footnote{The results in Theorems and 1 and 2 given for a predecessor-follower pair can be extended to platoons of arbitrary length, thanks to the disturbance decoupling property.}.
The initial conditions and powertrain time constants are shown in Table~\ref{table1}:
in addition, we use $\theta_1=1,\ \theta_2=1,\ h=0.72,\ Q=0.7I,\ \gamma_1=0.3,\ \gamma_2=0.01$ for the controller and $\tau_{\bar{f}} = 0.5$ as nominal powertrain time constant. 
Since $\phi$ is a scalar, we simply select one sample $\phi_0$ in (\ref{adaptive law CL}). 
\begin{table}[h]
    \caption{Actual powertrain constants and initial conditions}
    \label{table1}
		\centering
    \begin{tabular}{|c||c|c|c|c|}
    \hline
    $i$ & $\tau_i$ & $d_i(0)$ & $v_i(0)$ & $a_i(0)$\\
    \hline
    $0$ & 0.2 & 0 & 10 & 0\\
    \hline
    $1$ & 0.05 & -2 & 12 & 0\\
    \hline
    $2$ & 0.1 & -4 & 8 & 0\\
    \hline
    $3$ & 0.3 & -6 & 11 & 0\\
    \hline
    $4$ & 0.25 & -8 & 10 & 0\\
    \hline
    \end{tabular}
    \end{table}

\noindent In the following we show three controllers:
\begin{enumerate}
    \item a non-adaptive scenario with controller as in (\ref{input matrix}), but with  incorrect knowledge of $\tau_i,\ i\in\{0,1,2,3,4,5\}$ (we select them as $0.15$); 
    \item an adaptive scenario where the knowledge of $\tau_i,\ i\in\{0,1,2,3,4,5\}$ is not available, but estimated online with MRAC controller (\ref{actual controller}) and adaptive law (\ref{adaptive law});
    \item as in 2), the knowledge of $\tau_i,\ i\in\{0,1,2,3,4,5\}$ is not available, but estimated online with MRAC-CL controller (\ref{actual controller}) and adaptive law (\ref{adaptive law CL}).
\end{enumerate}
The three scenarios are considered in the presence and in the absence of PE. 
To simulate PE, we let the leader proceed with $u_0=\sin(0.1t)+0.5\sin(0.5t)$. 
To simulate absence of PE, we let the leader proceed with $u_0=0$. {\color{black}As mentioned at the beginning of Sect. IV, a platoon converging to zero acceleration (i.e., constant velocity) would fail to meet PE.} 

{\color{black}As it can be seen in Fig. 2(a), disturbance decoupling is not achieved by a linear control with incorrect knowledge} of $\tau_i,\ i\in\{1,2,3,4\}$: {\color{black}however, both adaptive strategies achieve disturbance decoupling}, as the spacing errors converge to zero despite the oscillating $u_0(\cdot)$. In Fig. 2(b) the spacing error converges to zero for all three control strategies {\color{black}for the reason that $u_0=0$, i.e., there is no disturbance.} The convergence of $\hat{\tau}_i,\ i\in\{1,2,3,4\}$ for MRAC and MRAC-CL is shown in Fig.~\ref{tau}. 
In the presence of PE, both MRAC and MRAC-CL converge to the actual $\tau_i$ (with MRAC-CL being faster); 
however, in the absence of PE, only MRAC-CL converges to the actual $\tau_i$, which is not the case for MRAC. When calculating the percentage estimation error from Fig.~\ref{tau}, i.e., $\left\|\tilde{\tau}_i\right\|/\left\|\tau_i\right\| \%$, we find that the proposed MRAC-CL has 4\% error as compared to 11\% of the standard MRAC. 

\section{Conclusions}
We proposed a new adaptive longitudinal platooning strategy in the framework of concurrent learning. 
The proposed strategy improves available platooning strategies since convergence to the true parameters can be achieved without requiring persistence of excitation in the vehicle behavior. 
{\color{black}
Interesting topics for future work are integral CL implementation in line with~\cite{c35}, output-feedback implementation in line with~\cite{c29}, and switched systems implementation with switching between CACC and ACC~\cite{c12}. 
}





\bibliographystyle{IEEEtran}
\bibliography{reference}
\end{document}